\documentclass[aps,pra,onecolumn,superscriptaddress,floatfix]{revtex4-2}

\usepackage{amsmath,amssymb,amsfonts,mathtools}
\usepackage{amsthm}
\usepackage{graphicx}
\usepackage[section]{placeins}
\usepackage{bm}
\usepackage[dvipsnames,svgnames]{xcolor}
\usepackage[colorlinks=true]{hyperref}
\hypersetup{
	linkcolor   = NavyBlue,
	citecolor   = BrickRed,
	urlcolor    = Teal,
	pdftitle    = {Boundary-Aware QFT Block-Encoding of Fractional Laplacians},
	pdfauthor   = {YJ},
	pdfkeywords = {fractional Laplacian, open boundary, quantum algorithms,
		block-encoding, QFT, circulant embedding, spectral oracle},
}
\usepackage{qcircuit}

\newcommand{\Fh}{\mathcal{F}_h}
\newcommand{\Fhinv}{\mathcal{F}_h^{-1}}

\newcommand{\AahN}{A^{(N)}_{\alpha,h}}

\newcommand{\lmax}{\lambda_{\max}}

\newcommand{\ket}[1]{|#1\rangle}
\newcommand{\bra}[1]{\langle #1|}
\newcommand{\diag}{\mathrm{diag}}

\newtheorem{theorem}{Theorem}

\newtheorem{remark}{Remark}
\newtheorem{proposition}{Proposition}

\begin{document}
	
	\title{Boundary-Aware QFT Block-Encoding of Fractional Laplacians}
	
	\author{Younes Javanmard}
	\email{javanmard.younes@gmail.com}
	% \affiliation{Institut f\"ur Theoretische Physik, Leibniz Universit\"at Hannover,
		% Appelstra\ss e 2, 30167 Hannover, Germany}
	\affiliation{Institut f\"ur Physik und Astronomie, Technische Universit\"at Berlin,
		Hardenbergstra\ss e 36, EW 7-1, 10623 Berlin, Germany}
	
	\author{Sina Kazemian}
	\email{skazemi5@uwo.ca}
	\affiliation{University of Western Ontario, Department of Physics and Astronomy, London, Ontario, Canada}
	
	\date{\today}
	
	\begin{abstract}
		We study the quantum Fourier transform (QFT) block-encoding of the semi-discrete fractional Laplacian on
		bounded domains with open, zero-extension boundary conditions. In the notation
		of the main construction, the target operator is the finite Toeplitz truncation
		\(A^{(N)}_{\alpha,h}\) obtained from the full-lattice semi-discrete operator
		with symbol \(|\xi|^\alpha\). A finite QFT register, however, diagonalizes
		circulant matrices rather than Toeplitz truncations. The native QFT circuit
		therefore implements a periodic surrogate \(\widetilde A^{(N)}_{\alpha,h}\),
		not the open-boundary operator. We identify this mismatch through an exact
		Toeplitz-to-circulant aliasing identity.
		
		To recover the open-boundary action, we zero-pad the state into a larger
		\(M\)-point QFT register, apply the same Fourier-symbol block-encoding, and
		compress back to the physical subspace. The resulting compressed block
		satisfies
		\(P_{N\to M}^{\dagger}\widetilde A^{(M)}_{\alpha,h}P_{N\to M}
		= A^{(N)}_{\alpha,h}+E^{(M)}\), where \(E^{(M)}\) is controlled by the tail
		of the semi-discrete convolution kernel. Thus the QFT layer implements the
		fractional symbol, while zero-padding supplies the open-boundary geometry. The
		construction is an operator-compilation primitive for boundary-aware quantum
		simulation rather than a complete PDE solver.
	\end{abstract}
	
	\maketitle
	
	\section{Introduction}
	
	Fractional Laplacians are simple in Fourier space and subtle in real space. On
	the full line or on a periodic domain, \((-\Delta)^{\alpha/2}\) acts by
	multiplication with the symbol \(|\xi|^\alpha\). On a bounded domain, however,
	the operator is nonlocal and boundary-sensitive: how the function is extended
	outside the physical interval is part of the finite-domain operator
	\cite{MetzlerKlafter2000,Applebaum2009,Laskin2002,DuoWykZhang2018,HuangOberman2014,
		Kwasnicki2017,CaffarelliSilvestre2007,BucurValdinoci2016}. Such boundary-sensitive nonlocal operators arise across anomalous diffusion,
	L\'evy transport, nonlocal heat conduction, and bounded-domain fractional
	models, where finite-size effects, exterior data, and long-range interactions
	can change the effective operator~\cite{defterli2015fractional,duo2017comparative,lischke2020fractional,dhar2019anomalous,kazemian2025electron,BonitoPasciak2015,AcostaBorthagaray2017,BonitoLeiPasciak2019}.
	
	This distinction matters for quantum algorithms. QFT circuits naturally
	implement Fourier multipliers, so \(|\xi|^\alpha\) is an attractive diagonal
	oracle. This Fourier-space strategy is closely related to recent QFT-based
	approaches to differential-operator and wave-equation simulation
	\cite{Lubasch2025,Wright2024}. But a finite QFT register diagonalizes
	circulants and therefore imposes periodic finite geometry unless a boundary
	mechanism is added. The
	open-boundary semi-discrete fractional Laplacian considered here is instead a
	Toeplitz truncation of an infinite-lattice convolution kernel. Thus the native
	QFT circuit block-encodes the periodic finite-domain operator rather than the
	open-boundary target.
	
	We use the semi-discrete fractional-Laplacian operator of
	Ref.~\cite{ZhouZhang2023}, whose full-lattice symbol is \(|\xi|^\alpha\). The
	open-boundary target \(\AahN\) is
	obtained by zero extension, application of the full-lattice operator, and
	restriction to the physical sites. The QFT-native operator
	\(\widetilde{A}^{(N)}_{\alpha,h}\), by
	contrast, is a circulant periodization of the same Toeplitz kernel. The aliasing
	identity shows exactly where the artificial wrap-around terms enter.
	Zero-padding the input into a larger QFT register moves those wrap-around
	couplings away from the physical support, leaving only a tail residual
	\(E^{(M)}\).
	
	This places the present work between classical Toeplitz/circulant embedding
	methods and quantum Fourier-space operator constructions. Classical FFT methods
	exploit circulant embeddings to apply Toeplitz matrices efficiently
	\cite{StrangCirculant1986,DavisCirculant1979,Gray2006,ChanNg1996}, whereas QFT-based quantum
	algorithms naturally implement diagonal Fourier multipliers. The point addressed
	here is the boundary interface between these two views: the symbol oracle is
	periodic by construction, while the zero-extension fractional Laplacian is not.
	Our contribution is to isolate this Toeplitz-versus-circulant obstruction and
	give a compressed QFT block whose boundary error is explicitly controlled.
	
	The result is not an end-to-end PDE solver. It is a boundary-aware
	block-encoding primitive that can be used by higher-level routines such as
	QSVT, Hamiltonian simulation, quantum linear-system solvers, and quantum
	algorithms for differential equations and PDEs
	\cite{harrow2009quantum,childs2017quantum,gilyen2019quantum,low2019hamiltonian,Cao2013,MontanaroPallister2016,BerryChildsOstranderWang2017,CostaJordanOstrander2019,ChildsLiuOstrander2021,MartynRossiTanChuang2021}.
	The remainder of the paper develops the operator distinction, the zero-padded
	compressed block, its tail-error bound, and numerical checks of the
	boundary-aliasing mechanism.
	
	\section{Operators and boundary mismatch}
	\label{sec:operators}
	
	Fix a mesh size \(h>0\) and \(0<\alpha\le2\). On the full lattice \(h\mathbb Z\),
	the semi-discrete fractional-Laplacian operator is the Fourier multiplier
	\begin{equation}
		A_{\alpha,h}=\Fhinv M_{|\xi|^\alpha}\Fh .
		\label{eq:A-full}
	\end{equation}
	Equivalently, it is convolution by a real even kernel \(c_m=\omega_{\alpha,h}(m)\).
	The detailed transform convention and kernel formula are collected in
	Appendix~\ref{app:kernel}.
	
	For the finite physical domain \(\Lambda_N=\{0,\ldots,N-1\}\), the open
	zero-extension target is
	\begin{equation}
		A_{\alpha,h}^{(N)}=R_NA_{\alpha,h}Z_N,
		\qquad
		(A_{\alpha,h}^{(N)})_{ij}=c_{i-j}.
		\label{eq:A-target}
	\end{equation}
	Here \(Z_N\) embeds a vector into the full lattice by setting exterior entries to
	zero, and \(R_N\) restricts back to \(\Lambda_N\). Thus \(\AahN\) is the finite
	Toeplitz matrix that represents the open-boundary fractional Laplacian.
	
	The construction is easiest to follow by keeping the three operators in
	Table~\ref{tab:operators} separate.
	
	\begin{table}[h]
		\centering
		\begin{tabular}{lll}
			\hline\hline
			Operator & Geometry & Role \\
			\hline
			\(A_{\alpha,h}\) & infinite lattice & full semi-discrete symbol operator \\
			\(\AahN\) & open finite interval & Toeplitz target from zero extension \\
			\(\widetilde{A}^{(N)}_{\alpha,h}\) & periodic finite cycle & QFT-native circulant surrogate \\
			\hline\hline
		\end{tabular}
		\caption{Operators appearing in the construction. The full-lattice operator
			carries the Fourier symbol, the Toeplitz truncation is the open-boundary target,
			and the circulant surrogate is the operator produced natively by a finite QFT
			register.}
		\label{tab:operators}
	\end{table}
	
	Let \(N=2^n\), and use the FFT-ordered frequency grid in
	Eq.~\eqref{eq:fft-grid},
	\begin{equation}
		\bar\xi_k =
		\begin{cases}
			2\pi k/(Nh), & 0\le k<N/2,\\
			2\pi k/(Nh)-2\pi/h, & N/2\le k<N.
		\end{cases}
		\label{eq:fft-grid}
	\end{equation}
	With this ordering, the finite QFT construction gives
	\begin{equation}
		\widetilde{A}^{(N)}_{\alpha,h}
		=
		QFT_N^{-1}\operatorname{diag}(|\bar\xi_k|^\alpha)QFT_N .
		\label{eq:A-circ}
	\end{equation}
	This matrix is circulant. It is the correct operator for periodic boundary
	geometry, but it is not the open-boundary Toeplitz target in Eq.~\eqref{eq:A-target}.
	
	\begin{proposition}[Toeplitz-to-circulant aliasing]
		\label{prop:alias}
		The entries of the QFT-native circulant surrogate satisfy
		\begin{equation}
			(\widetilde{A}^{(N)}_{\alpha,h})_{ij}
			=
			\sum_{\ell\in\mathbb Z}c_{i-j+\ell N}.
			\label{eq:aliasing}
		\end{equation}
		Therefore \(\widetilde{A}^{(N)}_{\alpha,h}-\AahN\) consists exactly of the
		nonzero periodic images \(\sum_{\ell\ne0}c_{i-j+\ell N}\).
	\end{proposition}
	
	Equation~\eqref{eq:aliasing} is the boundary obstruction in one line. The
	Toeplitz target uses the physical separation \(i-j\), whereas the circulant
	surrogate sums over all periodic images \(i-j+\ell N\). For matrix entries well
	inside the interval, these images are far away and are suppressed by the kernel
	tail. Near opposite corners, however, the periodic image can be close: the entry
	\((0,N-1)\) contains \(c_{-(N-1)+N}=c_1\), which is generally \(O(1)\). This is
	why the native QFT operator may behave well on bulk-localized states while still
	failing uniformly as an open-boundary operator.

	\section{QFT block-encoding and zero-padded compression}
	\label{sec:block}
	
	The diagonal spectral oracle stores the normalized symbol in an ancilla
	amplitude. With \(\lmax=(\pi/h)^\alpha\), define
	\begin{equation}
		\phi_k=\frac{|\bar\xi_k|^\alpha}{\lmax}.
		\label{eq:normalized-symbol}
	\end{equation}
	The ideal oracle uses the normalized symbol in Eq.~\eqref{eq:normalized-symbol}
	and acts as
	\begin{equation}
		U_D\ket{k}\ket{0}
		=
		\ket{k}\left(\phi_k\ket{0}+\sqrt{1-\phi_k^2}\ket{1}\right),
		\label{eq:diagonal-oracle}
	\end{equation}
	with any arithmetic work qubits uncomputed. The native block-encoding unitary is
	\begin{equation}
		U_{\rm BE}^{(N)}
		=
		(QFT_N^{-1}\otimes I)U_D(QFT_N\otimes I).
		\label{eq:native-unitary}
	\end{equation}
	Its ancilla block is
	\begin{equation}
		(I\otimes\langle0|)U_{\rm BE}^{(N)}(I\otimes|0\rangle)
		=
		\widetilde{A}^{(N)}_{\alpha,h}/\lambda_{\max}.
		\label{eq:native-block}
	\end{equation}
	This places the construction in the broader family of LCU, block-encoding, and
	quantum-signal-processing based simulation methods
	\cite{ChildsWiebe2012,BerryChildsCleveKothariSomma2015,LowChuang2017}.
	
	The native QFT and block-encoding circuits are shown in Figs.~\ref{fig:qft-circuit} and~\ref{fig:be-circuit}.
	The QFT circuit and its role as a spectral basis transformation are standard in
	quantum computation \cite{NielsenChuang2010}.
	
	\begin{figure}[t]
		\centering
		\includegraphics[width=0.62\linewidth]{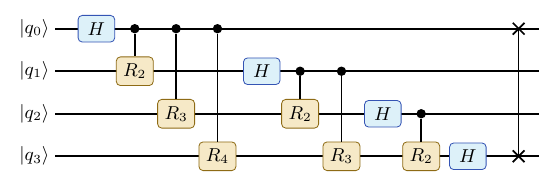}
		\caption{Illustrative four-qubit QFT decomposition. The finite QFT is the
			spectral layer used by the block-encoding, but by itself it imposes periodic
			rather than open geometry.}
		\label{fig:qft-circuit}
	\end{figure}
	
	\begin{figure}[t]
		\centering
		\includegraphics[width=0.72\linewidth]{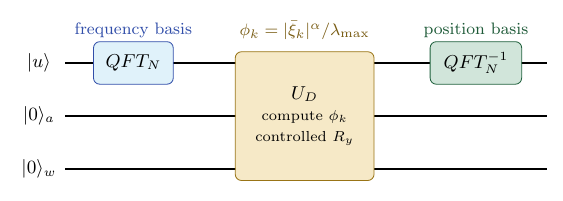}
		\caption{Native QFT block-encoding of the periodic circulant surrogate
			\(\widetilde{A}^{(N)}_{\alpha,h}\). The frequency oracle encodes the fractional symbol as an
			ancilla amplitude.}
		\label{fig:be-circuit}
	\end{figure}
	
	To recover the open-boundary action, choose \(M=2^m\) with \(M\ge2N\), let
	\(P_{N\to M}:\mathbb C^N\to\mathbb C^M\) pad a vector by zeros, and apply the
	same symbol-oracle construction on the larger register. The size-\(M\) circulant
	surrogate is denoted \(\widetilde{A}^{(M)}_{\alpha,h}\).
	
	\begin{theorem}[Zero-padded compressed block]
		\label{thm:embedding}
		Let \(M=2^m\) with \(M\ge 2N\). Then
		\begin{equation}
			P_{N\to M}^{\dagger}
			\widetilde{A}^{(M)}_{\alpha,h}
			P_{N\to M}
			=
			A_{\alpha,h}^{(N)}+E^{(M)},
			\label{eq:compressed-operator}
		\end{equation}
		where, for \(0\le i,j\le N-1\),
		\begin{equation}
			(E^{(M)})_{ij}
			=
			\sum_{\ell\neq0}c_{i-j+\ell M}.
			\label{eq:E-entries}
		\end{equation}
		Consequently, the compressed ancilla block of the size-\(M\) QFT
		block-encoding satisfies
		\begin{equation}
			(P_{N\to M}^{\dagger}\otimes\langle0|)
			U_{\rm BE}^{(M)}
			(P_{N\to M}\otimes|0\rangle)
			=
			\frac{A_{\alpha,h}^{(N)}+E^{(M)}}{\lambda_{\max}}.
			\label{eq:compressed-block}
		\end{equation}
	\end{theorem}
	
	This is a compressed block on the zero-padded logical subspace. If it is
	rewritten as a strict block-encoding of an \(n\)-qubit logical operator, the
	padding qubits and arithmetic work qubits must be included in the ancilla count.
	
	\begin{figure}[t]
		\centering
		\includegraphics[width=0.82\linewidth]{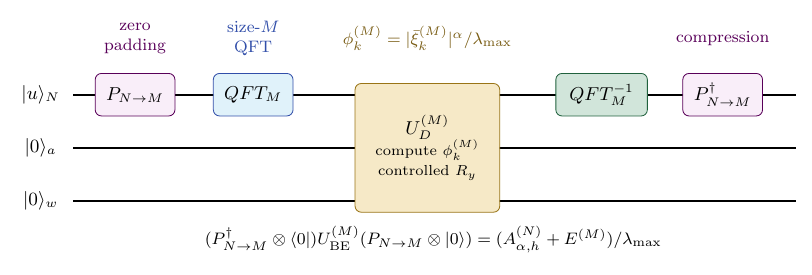}
		\caption{Compressed open-boundary block from zero-padding. The size-\(M\)
			QFT circuit is periodic on the padded register, while compression back to the
			first \(N\) coordinates gives \((\AahN+E^{(M)})/\lmax\).}
		\label{fig:embedding-circuit}
	\end{figure}
	
	\begin{remark}[Exact \(2N\)-Toeplitz embedding]
		The classical Toeplitz-FFT construction can be viewed as a different quantum
		route \cite{chan1988optimal,wu2026high}. With the convention
		\begin{equation}
			(C_{2N})_{ij}=g_{(i-j)\bmod 2N},
			\label{eq:toeplitz-embedding-convention}
		\end{equation}
		one can choose the generator in Eq.~\eqref{eq:toeplitz-embedding-generator},
		\begin{equation}
			g=(c_0,c_1,\ldots,c_{N-1},0,c_{-(N-1)},\ldots,c_{-1})
			\label{eq:toeplitz-embedding-generator}
		\end{equation}
		so that the compression identity
		\begin{equation}
			P_{N\to 2N}^{\dagger}C_{2N}P_{N\to 2N}=A_{\alpha,h}^{(N)}
			\label{eq:exact-2N-embedding}
		\end{equation}
		holds exactly. Since \(C_{2N}\) is circulant, it is diagonalized by the \(2N\)-point
		Fourier transform. A quantum implementation of this exact embedding would
		therefore require an oracle for the finite-circulant eigenvalues \(\gamma_k\),
		normalized by some \(\Lambda\ge\max_k|\gamma_k|\). This places the exact
		Toeplitz-embedding route closer to a finite-operator compilation problem, where
		one may use LCU-style or structured-operator interfaces rather than a direct
		analytic symbol oracle \cite{javanmard2026coefficient}. This is not the same
		oracle as the simple symbol sample \(|\bar\xi_k|^\alpha/\lambda_{\max}\). The
		construction in this paper keeps the simple symbol oracle and pays instead a
		controllable kernel-tail residual.
	\end{remark}
	
	\section{Error bound and resource scaling}
	\label{sec:error}
	
	The residual in Eq.~\eqref{eq:E-entries} consists only of kernel images whose
	distance from the physical support is at least \(M-N+1\). The Schur-test estimate
	proved in Appendix~\ref{app:proofs} gives
	\begin{equation}
		\|E^{(M)}\|_2
		\le
		\sum_{|r|\ge M-N+1}|c_r|.
		\label{eq:tail-bound}
	\end{equation}
	Using the kernel decay in Eq.~\eqref{eq:kernel-decay-main},
	\begin{equation}
		|c_r|\le C_\alpha h^{-\alpha}|r|^{-r_\alpha},
		\qquad
		r_\alpha=\min(2,1+\alpha),
		\label{eq:kernel-decay-main}
	\end{equation}
	one obtains the sufficient scaling
	\begin{equation}
		M=O\!\left(N+\epsilon^{-1/(r_\alpha-1)}\right),
		\qquad
		r_\alpha=\min(2,1+\alpha).
		\label{eq:M-scaling}
	\end{equation}
	For \(\alpha=1\), this gives \(M-N=O(\epsilon^{-1})\); for
	\(0<\alpha<1\), the heavier tail gives \(M-N=O(\epsilon^{-1/\alpha})\).
	
	The gate count per block-encoding call separates into QFT cost, reversible
	symbol arithmetic, and rotation synthesis:
	\begin{equation}
		G_{\rm BE}
		=
		O(m^2)+C_\alpha(m,b)+C_{\rm rot}(b),
		\qquad
		M=2^m.
		\label{eq:GBE}
	\end{equation}
	Here \(b\) is the arithmetic precision, \(C_\alpha(m,b)\) is the cost of
	evaluating the folded normalized symbol to \(b\) bits, and \(C_{\rm rot}(b)\) is
	the cost of the controlled rotation. For fixed rational \(\alpha\), standard
	reversible arithmetic gives \(C_\alpha(m,b)\) polynomial in \(m\) and \(b\)
	\cite{Haener2018}. The classical \(O(2N\log(2N))\) FFT cost refers to applying a
	circulant to an explicitly stored vector; it is not the gate count of the QFT
	block-encoding, whose Fourier layer is logarithmic in \(M\).
	
	Finite arithmetic, rotation synthesis, and approximate-QFT errors add to the
	boundary residual. At the level of this operator primitive, the total
	implementation error has the form shown in Eq.~\eqref{eq:total-implementation-error},
	\begin{equation}
		O(\epsilon^{(M)}+\epsilon_{\rm arith}+\epsilon_{\rm rot}+\epsilon_{\rm QFT}).
		\label{eq:total-implementation-error}
	\end{equation}
	
	\subsection{Three-dimensional tensor-product extension at the operator-identity level}
	
	The theorem above is stated for the one-dimensional operator, where the
	Toeplitz-versus-circulant obstruction is most transparent. At the
	operator-identity level, the same padding-and-compression mechanism extends to
	cubic tensor-product grids with open zero-extension boundaries. The
	grid and multi-index notation are fixed in Eq.~\eqref{eq:3d-grid}:
	\begin{equation}
		\Lambda_N^3=\{0,\ldots,N-1\}^3,
		\qquad
		\mathbf i=(i_x,i_y,i_z),
		\label{eq:3d-grid}
	\end{equation}
	Let \(A_{\alpha,h}^{(N,3)}\) denote the restriction of the full-lattice
	three-dimensional semi-discrete operator to \(\Lambda_N^3\) after zero
	extension. Its entries have the block-Toeplitz form in
	Eq.~\eqref{eq:3d-target}:
	\begin{equation}
		(A_{\alpha,h}^{(N,3)})_{\mathbf i,\mathbf j}
		=
		c_{\mathbf i-\mathbf j}.
		\label{eq:3d-target}
	\end{equation}
	The finite QFT-native operator is defined by Eq.~\eqref{eq:3d-circ},
	\begin{equation}
		\widetilde{A}^{(N,3)}_{\alpha,h}
		=
		(QFT_N^{\otimes 3})^\dagger
		D_{\alpha,h}^{(N,3)}
		QFT_N^{\otimes 3},
		\label{eq:3d-circ}
	\end{equation}
	with diagonal symbol matrix
	\begin{equation}
		D_{\alpha,h}^{(N,3)}
		=
		\operatorname{diag}
		\left[
		\left(
		|\bar\xi_{k_x}|^2+
		|\bar\xi_{k_y}|^2+
		|\bar\xi_{k_z}|^2
		\right)^{\alpha/2}
		\right].
		\label{eq:3d-symbol}
	\end{equation}
	This operator is periodic in all three directions. Its entries satisfy the
	multidimensional aliasing identity in Eq.~\eqref{eq:3d-aliasing}:
	\begin{equation}
		(\widetilde{A}^{(N,3)}_{\alpha,h})_{\mathbf i,\mathbf j}
		=
		\sum_{\boldsymbol{\ell}\in\mathbb Z^3}
		c_{\mathbf i-\mathbf j+N\boldsymbol{\ell}}.
		\label{eq:3d-aliasing}
	\end{equation}
	Thus the native three-dimensional QFT circuit introduces wrap-around couplings
	across faces, edges, and corners of the cube.
	
	To recover the open-boundary action, pad each spatial direction from \(N\) to
	\(M\). With \(P_{N\to M}^{(3)}\) denoting the zero-padding isometry from
	\(\mathbb C^{N^3}\) into \(\mathbb C^{M^3}\), the compressed operator is
	Eq.~\eqref{eq:3d-compressed}:
	\begin{equation}
		(P_{N\to M}^{(3)})^\dagger
		\widetilde{A}^{(M,3)}_{\alpha,h}
		P_{N\to M}^{(3)}
		=
		A_{\alpha,h}^{(N,3)}+E^{(M,3)},
		\label{eq:3d-compressed}
	\end{equation}
	where the residual entries are
	\begin{equation}
		(E^{(M,3)})_{\mathbf i,\mathbf j}
		=
		\sum_{\boldsymbol{\ell}\in\mathbb Z^3\setminus\{\mathbf 0\}}
		c_{\mathbf i-\mathbf j+M\boldsymbol{\ell}}.
		\label{eq:3d-residual}
	\end{equation}
	The same row/column-sum argument gives the multidimensional tail bound in
	Eq.~\eqref{eq:3d-tail-bound}:
	\begin{equation}
		\|E^{(M,3)}\|_2
		\le
		\sum_{\mathbf r:\,\|\mathbf r\|_\infty\ge M-N+1}
		|c_{\mathbf r}|.
		\label{eq:3d-tail-bound}
	\end{equation}
	The convergence rate in three dimensions is therefore determined by the
	corresponding multidimensional kernel tail. The normalization for the isotropic
	symbol is Eq.~\eqref{eq:3d-lambda-max}:
	\begin{equation}
		\lambda_{\max}^{(3D)}
		=
		\left(\sqrt{3}\,\pi/h\right)^\alpha.
		\label{eq:3d-lambda-max}
	\end{equation}
	Thus the 3D circuit uses \(3m\) data qubits for an \(M^3\) padded grid, one
	block-encoding ancilla, and the arithmetic work qubits required to evaluate the
	normalized symbol in Eq.~\eqref{eq:3d-normalized-symbol}:
	\begin{equation}
		\frac{
			\left(
			|\bar\xi_{k_x}|^2+
			|\bar\xi_{k_y}|^2+
			|\bar\xi_{k_z}|^2
			\right)^{\alpha/2}
		}{
			\lambda_{\max}^{(3D)}
		}.
		\label{eq:3d-normalized-symbol}
	\end{equation}
	The QFT layer consists of three exact one-dimensional QFTs and therefore has
	cost \(O(3m^2)\), plus the reversible arithmetic and rotation-synthesis costs.
	As in the one-dimensional construction, it is exact for the padded periodic
	surrogate and approximate for the open-boundary target through the
	multidimensional kernel-tail residual.
	
	\section{Numerical illustration}
	\label{sec:numerics}
	
	The numerical checks are not PDE-solver benchmarks. They test the operator
	identities above. Figure~\ref{fig:heatmaps} visualizes the matrix-level
	aliasing structure. Figure~\ref{fig:functional-benchmark} compares the operator
	action on simple bounded-domain functions, using reference functional tests.
	Figure~\ref{fig:scaling} confirms the decay of the zero-padding residual.
	
	For \(N=64\), \(h=1\), and \(\alpha=1.5\), the open Toeplitz matrix and the
	native circulant agree near the diagonal but differ near opposite corners. This
	is the wrap-around structure predicted by Eq.~\eqref{eq:aliasing}. For
	\(\alpha=1\), \(h=1\), the coefficient formula in
	Eq.~\eqref{eq:alpha-one-kernel}
	\(c_m=((-1)^m-1)/(\pi m^2)\) for \(m\ne0\) makes the corner effect explicit:
	\begin{equation}
		(\AahN)_{0,N-1}=c_{-(N-1)},\qquad
		(\widetilde{A}^{(N)}_{\alpha,h})_{0,N-1}=c_{-(N-1)}+c_1+\cdots .
		\label{eq:corner-aliasing}
	\end{equation}
	The second term in Eq.~\eqref{eq:corner-aliasing} contains the artificial
	coupling \(c_1=-2/\pi\).
	
	\begin{figure}[h]
		\centering
		\includegraphics[width=0.92\linewidth]{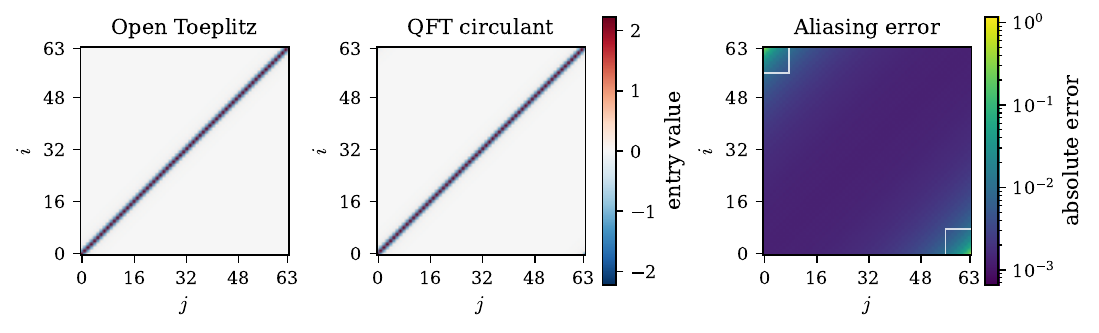}
		\caption{Matrix structure for \(N=64\), \(h=1\), and \(\alpha=1.5\). The
			third panel shows the absolute aliasing error on a logarithmic color scale;
			the highlighted corners are the wrap-around regions.}
		\label{fig:heatmaps}
	\end{figure}
	
	Figure~\ref{fig:functional-benchmark} gives a direct functional comparison
	using three bounded-domain test functions also used in
	Ref.~\cite{ZhouZhang2023}. We compare the open-boundary Toeplitz target, the
	native QFT-circulant surrogate, and the zero-padded compressed operator. The
	top row overlays the operator outputs, while the bottom row shows the pointwise
	absolute error on a logarithmic scale. The native QFT curve plays the role of a
	periodic Fourier method and displays boundary artifacts, whereas the padded
	construction follows the open-boundary target.
	
	\begin{figure}[t]
		\centering
		\includegraphics[width=0.98\linewidth]{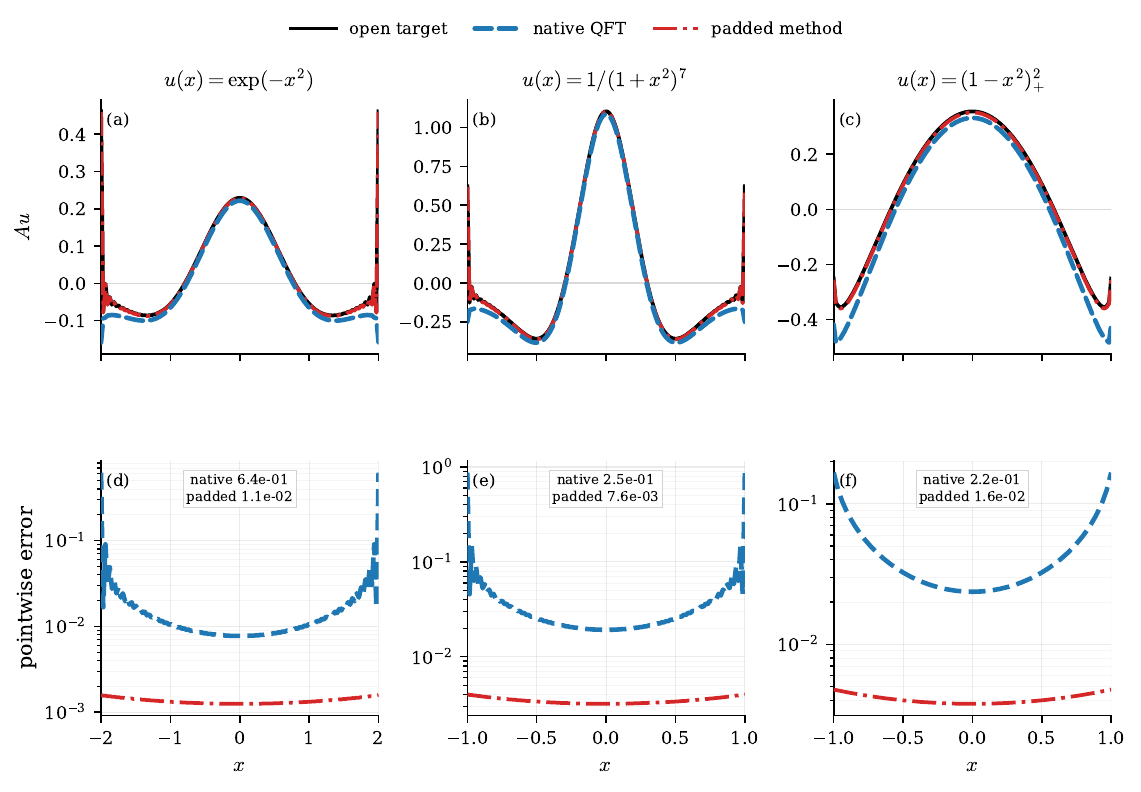}
		\caption{
			Functional comparison using three bounded-domain reference test functions.
			The top row compares the open-boundary Toeplitz action
			\(A^{(N)}_{\alpha,h}u\), the native QFT-circulant action
			\(\widetilde A^{(N)}_{\alpha,h}u\), and the zero-padded compressed action
			\(P_{N\to M}^{\dagger}\widetilde A^{(M)}_{\alpha,h}P_{N\to M}u\).
			The bottom row shows the corresponding pointwise absolute errors on a
			logarithmic scale. The native QFT surrogate follows periodic geometry and
			develops boundary artifacts, while the padded construction tracks the
			open-boundary target.
		}
		\label{fig:functional-benchmark}
	\end{figure}
	
	\begin{figure}[t]
		\centering
		\includegraphics[width=0.86\linewidth]{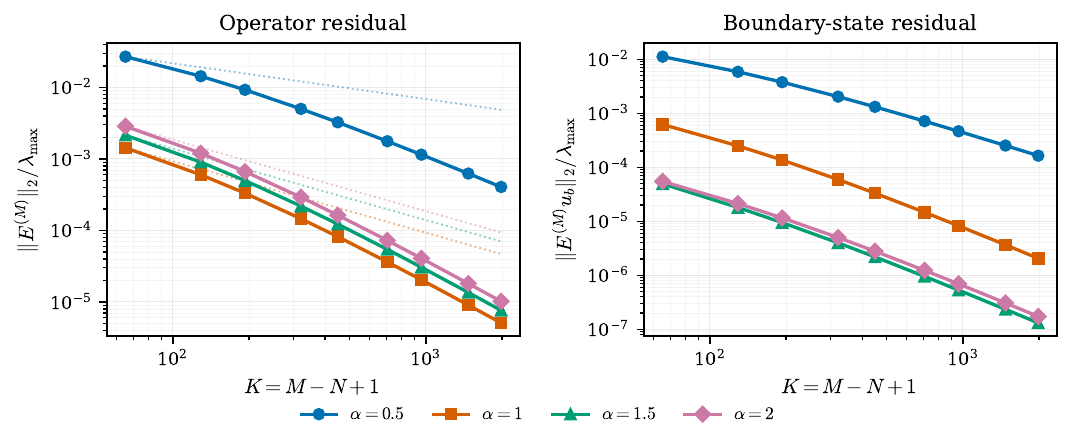}
		\caption{Decay of the zero-padding residual. Both panels are normalized by
			\(\lambda_{\max}\). Thin dotted lines show the predicted tail-bound slope. The left panel reports the dense-matrix spectral norm of \(E^{(M)}\), while
			the right panel reports \(\|E^{(M)}u_b\|_2/\lambda_{\max}\) for a normalized
			boundary-localized Gaussian \(u_b\) centered at \(j\) with width
			\(\sigma\).}
		\label{fig:scaling}
	\end{figure}
	
	\FloatBarrier
	
	\section{Discussion and conclusion}
	\label{sec:conclusion}
	
	The obstruction studied here is algebraic rather than algorithmic: finite QFT
	registers diagonalize circulants, while the open zero-extension fractional
	Laplacian is a Toeplitz truncation. The native QFT block-encoding of the
	fractional symbol is therefore exact for the periodic surrogate
	\(\widetilde{A}^{(N)}_{\alpha,h}\), but not for the open-boundary target
	\(A_{\alpha,h}^{(N)}\). The difference is precisely the periodized-image
	coupling in Eq.~\eqref{eq:aliasing}.
	
	Zero-padding supplies a boundary adapter. By embedding the \(N\)-site state into
	an \(M\)-site periodic register, applying the same QFT symbol block-encoding,
	and compressing back, one obtains the open-boundary Toeplitz action plus the
	kernel-tail residual \(E^{(M)}\). The residual is controlled by the padded
	separation \(M-N\), while the fractional-symbol oracle itself is unchanged.
	
	At the resource level, the Fourier layer costs \(O(m^2)\) gates for an exact
	\(M=2^m\) QFT, plus reversible arithmetic and rotation synthesis for the symbol
	oracle. This should be distinguished from the classical \(O(2N\log(2N))\)
	Toeplitz-FFT matrix-vector cost, which acts on an explicitly stored vector. The
	present construction is therefore best viewed as a boundary-aware
	operator-compilation primitive for quantum algorithms that use QFT-based
	fractional-Laplacian block-encodings. State preparation, success amplification,
	conditioning of higher-level solvers, and observable estimation are outside the
	scope of this operator-level construction. These higher-level questions are
	complementary to tensor-network-assisted quantum-search and
	Hamiltonian-simulation workflows, where the main difficulty lies in state
	preparation, search, or variational structure rather than in the
	boundary-aware compilation of a fixed operator oracle
	\cite{javanmard2024integrating}.
	
	The tensor-product-grid extension follows by applying the same padding and
	compression in each spatial direction, with the precision governed by the
	corresponding multidimensional kernel tail. Other boundary conditions would
	require different extension and compression maps.
	
	\begin{acknowledgments}
		The authors would like to dedicate this work to the memory of Dr. Saeed Gholibeiglou, a dear
		friend, colleague, and physicist whose curiosity, kindness, and friendship are
		remembered with deep respect.
	\end{acknowledgments}
	
	\newpage
	
	\appendix
	
	\section{Additional numerical diagnostics}
	\label{app:numerics}
	
	Figure~\ref{fig:gaussian-diagnostics} contrasts the action of the open target,
	the native periodic surrogate, and the zero-padded compressed operator on two
	Gaussian inputs. The bulk-localized state is nearly insensitive to the boundary
	model, while the boundary-localized state exposes the wrap-around tail of the
	native circulant. Figure~\ref{fig:center-sweep} shows the same effect as the
	Gaussian center is moved away from the boundary.
	
	\begin{figure}[h]
		\centering
		\includegraphics[width=0.92\linewidth]{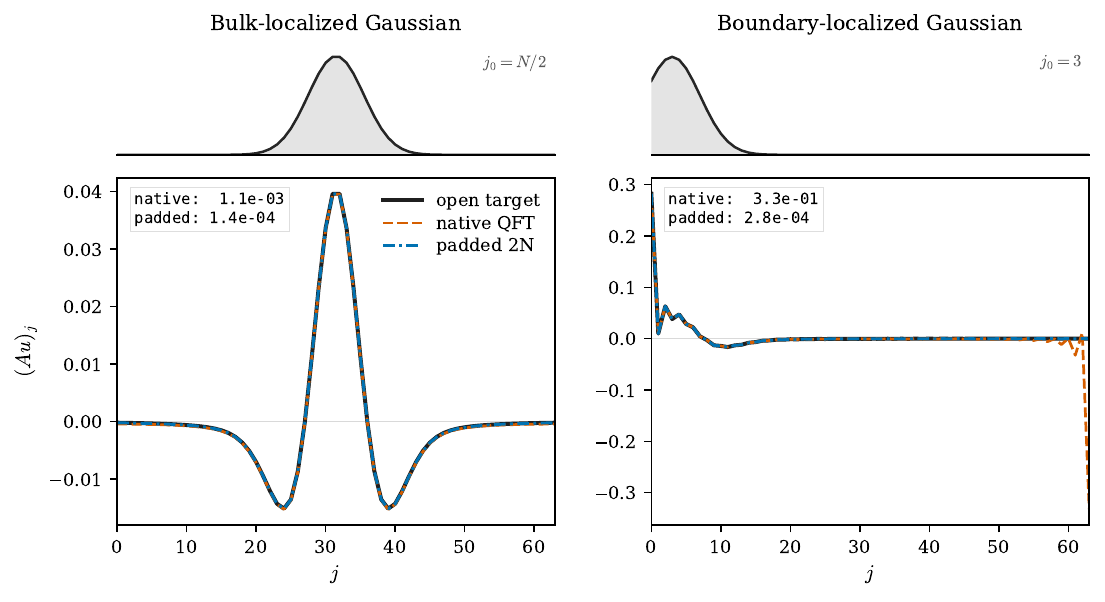}
		\caption{Gaussian-state diagnostics for \(N=64\), \(h=1\),
			\(\alpha=1.5\), and \(\sigma=4\). The panels compare the open-boundary
			target \(A^{(N)}_{\alpha,h}\), the native QFT-circulant surrogate
			\(\widetilde{A}^{(N)}_{\alpha,h}\), and the zero-padded embedding
			\(P^\dagger\widetilde{A}^{(2N)}_{\alpha,h}P\).}
		\label{fig:gaussian-diagnostics}
	\end{figure}
	
	\begin{figure}[h]
		\centering
		\includegraphics[width=0.86\linewidth]{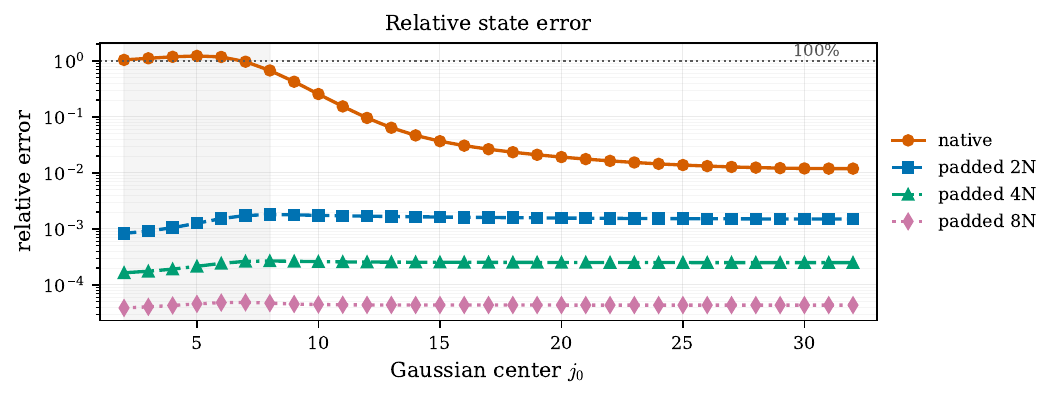}
		\caption{Relative state error as the Gaussian center \(j_0\) moves away from
			the boundary. The plotted quantity is
			\(\|({\rm surrogate}-A_{\alpha,h}^{(N)})u^{(j_0)}\|_2/
			\|A_{\alpha,h}^{(N)}u^{(j_0)}\|_2\).}
		\label{fig:center-sweep}
	\end{figure}
	
	\section{Semi-discrete kernel and decay}
	\label{app:kernel}
	
	For \(u\in\ell^2(h\mathbb Z)\), the semi-discrete Fourier transform convention
	is Eq.~\eqref{eq:semi-discrete-transform},
	\begin{equation}
		\widehat u_h(\xi)=\Fh[u](\xi)=h\sum_{j\in\mathbb Z}u_j e^{-i\xi hj},
		\qquad
		\xi\in[-\pi/h,\pi/h],
		\label{eq:semi-discrete-transform}
	\end{equation}
	with inverse Eq.~\eqref{eq:semi-discrete-inverse},
	\begin{equation}
		u_j=\frac{1}{2\pi}\int_{-\pi/h}^{\pi/h}\widehat u_h(\xi)e^{i\xi hj}\,d\xi .
		\label{eq:semi-discrete-inverse}
	\end{equation}
	The operator in Eq.~\eqref{eq:A-full} has the convolution form
	Eq.~\eqref{eq:kernel-convolution},
	\begin{equation}
		(A_{\alpha,h}u)_k=\sum_{j\in\mathbb Z}c_{k-j}u_j,
		\label{eq:kernel-convolution}
	\end{equation}
	where the kernel is given by Eq.~\eqref{eq:kernel-integral},
	\begin{equation}
		c_m=\omega_{\alpha,h}(m)
		=
		\frac{h}{2\pi}
		\int_{-\pi/h}^{\pi/h}|\xi|^\alpha e^{i\xi hm}\,d\xi
		=
		\frac{h}{\pi}\int_0^{\pi/h}\xi^\alpha\cos(mh\xi)\,d\xi .
		\label{eq:kernel-integral}
	\end{equation}
	After \(s=h\xi\), this becomes Eq.~\eqref{eq:scaled-kernel-integral},
	\begin{equation}
		c_m=h^{-\alpha}\frac{1}{\pi}\int_0^\pi s^\alpha\cos(ms)\,ds .
		\label{eq:scaled-kernel-integral}
	\end{equation}
	Standard Fourier-coefficient estimates for this piecewise-smooth even symbol
	give the decay estimate in Eq.~\eqref{eq:kernel-decay},
	\begin{equation}
		|c_m|\le C_\alpha h^{-\alpha}|m|^{-r_\alpha},
		\qquad
		r_\alpha=\min(2,1+\alpha),
		\qquad |m|\ge1.
		\label{eq:kernel-decay}
	\end{equation}
	For \(\alpha=1\) and \(h=1\), Eq.~\eqref{eq:alpha-one-kernel} gives
	\begin{equation}
		c_m=\frac{(-1)^m-1}{\pi m^2},\qquad m\ne0,\qquad c_0=\frac{\pi}{2}.
		\label{eq:alpha-one-kernel}
	\end{equation}
	For \(\alpha=2\), this construction corresponds to the semi-discrete spectral
	operator with continuum symbol \(|\xi|^2\) sampled on the semi-discrete Fourier
	cell. This should not be confused with the local second-difference Laplacian,
	whose finite-difference symbol is
	\begin{equation}
		4h^{-2}\sin^2(h\xi/2).
	\end{equation}

	\section{Proofs of aliasing identity and Schur bound}
	\label{app:proofs}
	
	\begin{proof}[Proof of Proposition~\ref{prop:alias}]
		The Toeplitz identity \((\AahN)_{ij}=c_{i-j}\) follows directly from the
		zero-extension definition. For the circulant surrogate, its first column is the
		inverse discrete Fourier transform of the sampled symbol. Inserting the
		semi-discrete kernel representation and using the finite Fourier comb
		in Eq.~\eqref{eq:finite-fourier-comb},
		\begin{equation}
			\frac{1}{N}\sum_{k=0}^{N-1}e^{2\pi ikq/N}
			=
			\mathbf 1_{q\in N\mathbb Z}
			\label{eq:finite-fourier-comb}
		\end{equation}
		folds the infinite Toeplitz sequence modulo \(N\). Hence
		Eq.~\eqref{eq:aliasing-proof-line} follows:
		\begin{equation}
			(\widetilde{A}^{(N)}_{\alpha,h})_{ij}
			=
			\sum_{\ell\in\mathbb Z}c_{i-j+\ell N}.
			\label{eq:aliasing-proof-line}
		\end{equation}
		Subtracting \(c_{i-j}\) gives the nonzero-image error.
	\end{proof}
	
	\begin{proof}[Proof of Theorem~\ref{thm:embedding} and Eq.~\eqref{eq:tail-bound}]
		Applying Proposition~\ref{prop:alias} at size \(M\) gives
		Eq.~\eqref{eq:size-M-aliasing-proof},
		\begin{equation}
			(\widetilde{A}^{(M)}_{\alpha,h})_{ij}
			=
			\sum_{\ell\in\mathbb Z}c_{i-j+\ell M},
			\qquad 0\le i,j\le M-1.
			\label{eq:size-M-aliasing-proof}
		\end{equation}
		Restricting both indices to \(0\le i,j\le N-1\) separates the \(\ell=0\) term,
		which is \((\AahN)_{ij}\), from the residual in Eq.~\eqref{eq:E-entries}. This
		proves Eq.~\eqref{eq:compressed-operator}; the compressed block
		Eq.~\eqref{eq:compressed-block} follows by inserting the exact ancilla block of
		the size-\(M\) block-encoding.
		
		For the norm bound, fix \(i\). Since \(0\le i,j\le N-1\) and \(\ell\ne0\),
		Eq.~\eqref{eq:padded-distance} holds:
		\begin{equation}
			|i-j+\ell M|\ge M-N+1.
			\label{eq:padded-distance}
		\end{equation}
		For \(M\ge2N\), the map \((j,\ell)\mapsto i-j+\ell M\) is injective at fixed
		\(i\), so each row sum is bounded by Eq.~\eqref{eq:row-sum-bound}:
		\begin{equation}
			\sum_{j=0}^{N-1}|(E^{(M)})_{ij}|
			\le
			\sum_{|r|\ge M-N+1}|c_r|.
			\label{eq:row-sum-bound}
		\end{equation}
		The same bound holds for column sums because \(c_r=c_{-r}\). The Schur test then
		gives Eq.~\eqref{eq:schur-bound-proof}:
		\begin{equation}
			\|E^{(M)}\|_2
			\le
			\sqrt{\|E^{(M)}\|_1\|E^{(M)}\|_\infty}
			\le
			\sum_{|r|\ge M-N+1}|c_r|.
			\label{eq:schur-bound-proof}
		\end{equation}
	\end{proof}
	
	Combining this with the decay estimate from Appendix~\ref{app:kernel}, with
	\(K=M-N+1\), yields Eq.~\eqref{eq:tail-sum-estimate}:
	\begin{equation}
		\sum_{|r|\ge K}|c_r|
		\le
		2C_\alpha h^{-\alpha}\sum_{r\ge K}r^{-r_\alpha}
		\le
		\frac{2C_\alpha h^{-\alpha}}{r_\alpha-1}K^{-(r_\alpha-1)}.
		\label{eq:tail-sum-estimate}
	\end{equation}
	Inverting Eq.~\eqref{eq:tail-sum-estimate} gives the scaling in
	Eq.~\eqref{eq:M-scaling}.
	
	\section{Concrete N=8 oracle example}
	\label{app:toy-example}
	
	Take \(N=8\), \(h=1\), and \(\alpha=1\). The FFT-ordered frequencies are listed
	in Eq.~\eqref{eq:N8-frequencies}:
	\begin{equation}
		0,\ \pi/4,\ \pi/2,\ 3\pi/4,\ -\pi,\ -3\pi/4,\ -\pi/2,\ -\pi/4 .
		\label{eq:N8-frequencies}
	\end{equation}
	With \(\lmax=\pi\), the normalized symbol is Eq.~\eqref{eq:N8-normalized-symbol}:
	\begin{equation}
		\phi_k
		=
		\frac{|\bar\xi_k|}{\pi}
		=
		\left(0,\frac14,\frac12,\frac34,1,\frac34,\frac12,\frac14\right).
		\label{eq:N8-normalized-symbol}
	\end{equation}
	Thus the diagonal oracle is Eq.~\eqref{eq:N8-oracle}:
	\begin{equation}
		U_D=\sum_{k=0}^{7}\ket{k}\bra{k}\otimes R_y(2\arccos\phi_k).
		\label{eq:N8-oracle}
	\end{equation}
	The ancilla-\(\ket{0}\) block is Eq.~\eqref{eq:N8-oracle-block}:
	\begin{equation}
		(I\otimes\bra{0})U_D(I\otimes\ket{0})
		=
		\diag\left(0,\frac14,\frac12,\frac34,1,\frac34,\frac12,\frac14\right).
		\label{eq:N8-oracle-block}
	\end{equation}
	Conjugating by \(QFT_8\) gives the periodic size-\(8\) surrogate in
	Eq.~\eqref{eq:N8-circulant}:
	\begin{equation}
		\widetilde{A}^{(8)}_{1,1}
		=
		QFT_8^{-1}
		\diag\left(0,\frac{\pi}{4},\frac{\pi}{2},\frac{3\pi}{4},
		\pi,\frac{3\pi}{4},\frac{\pi}{2},\frac{\pi}{4}\right)
		QFT_8 .
		\label{eq:N8-circulant}
	\end{equation}
	
	For the same parameters, the open Toeplitz corner entry is
	\((A_{1,1}^{(8)})_{0,7}=c_{-7}\). The circulant entry contains the image
	shown in Eq.~\eqref{eq:N8-corner}:
	\begin{equation}
		(\widetilde{A}^{(8)}_{1,1})_{0,7}=c_{-7}+c_1+c_9+\cdots,
		\label{eq:N8-corner}
	\end{equation}
	so the large contribution \(c_1=-2/\pi\) is the artificial wrap-around coupling
	that padding suppresses.
	
	\bibliographystyle{apsrev4-2}
	\bibliography{references}
	
\end{document}